\documentclass{amsart}

\usepackage[margin=1.4in]{geometry}
\usepackage{enumerate}
\usepackage{amssymb}
\usepackage{amsaddr}
\usepackage[colorinlistoftodos]{todonotes}
\usepackage[colorlinks]{hyperref}
\usepackage[T1]{fontenc}
\usepackage[utf8]{inputenc}
\usepackage{float}

\usepackage{xcolor}
\usepackage[boxed]{algorithm2e}

\usepackage[
    size=footnotesize,
    width=0.9\textwidth,
    labelfont={bf,sc}
]{caption}
\usepackage[
    size=footnotesize,
    labelfont={sc},
    width=0.9\textwidth
]{subcaption}

\usepackage{tikz}
\usetikzlibrary{quantikz}

\usepackage{nameref, cleveref}

\newtheorem{theorem}{Theorem}
\newtheorem{lemma}[theorem]{Lemma}
\newtheorem{proposition}[theorem]{Proposition}
\newtheorem{corollary}[theorem]{Corollary}

\theoremstyle{definition}
\newtheorem{definition}[theorem]{Definition}

\DeclarePairedDelimiter{\ceil}{\lceil}{\rceil}
\DeclarePairedDelimiter{\floor}{\lfloor}{\rfloor}
\DeclarePairedDelimiter{\sprod}{\langle}{\rangle}
\DeclarePairedDelimiter{\abs}{|}{|}

\DeclareMathOperator{\CX}{CX}
\DeclareMathOperator{\Phase}{P}
\DeclareMathOperator{\len}{len}

\DeclareMathOperator{\sgn}{sgn}

\DeclareMathOperator{\SPA}{SPA}
\DeclareMathOperator{\NPA}{NPA}
\DeclareMathOperator{\WPA}{WPA}
\DeclareMathOperator{\GRAY}{GRAY}
\DeclareMathOperator{\SWAP}{SWAP}

\makeatletter
\let\card\abs
\makeatother

\title[Diagonal operator decomposition on restricted topologies]{Diagonal operator decomposition\\on restricted topologies\\via enumeration of quantum state subsets\vspace{-1em}}
\author{Jan Tułowiecki$^{1,2}$, Łukasz Czerwiński, Konrad Deka$^{1,3}$, Jan Gwinner, Witold Jarnicki$^1$, Adam Szady$^1$}
\address{$^1$BEIT, Kraków, Poland \\
$^2$Department of Theoretical Computer Science, Jagiellonian
University, Kraków, Poland \\
$^3$Institute of Mathematics, Jagiellonian University, Kraków, Poland \\
\textup{\texttt{\{j.tulowiecki, konrad, witek, adam\}@beit.tech} }
\\[8pt]}

\begin{document}

\maketitle
\begin{abstract}
Various quantum algorithms require usage of arbitrary diagonal operators as subroutines. For their execution on a physical hardware, those operators must be first decomposed into target device's native gateset and its qubit connectivity for entangling gates. Here, we assume that the allowed gates are exactly the $\CX$ gate and the parameterized phase gate. We introduce a framework for the analysis of $\CX$-only circuits and through its lens provide solution constructions for several different device topologies (fully-connected, linear and circular). We also introduce two additional variants of the problem. Those variants can be used in place of exact decomposition of the diagonal operator when the circuit following it satisfies a set of prerequisites, enabling further reduction in the $\CX$ cost of implementation. Finally, we discuss how to exploit the framework for the decomposition of a particular, rather than general, diagonal operator.
\end{abstract}

\section{Introduction}

To execute quantum algorithms on physical devices, the high-level description of the quantum circuit
must be decomposed into the (quite limited) hardware native gateset. 
The quality of a decomposition is measured by the number of gates comprising the circuit - the lower the better.
The need for such decompositions has led researchers to formulate and solve various, either more general-purpose
or more specialized variants and aspects of the problem.
The task may depend on the circuit that needs to be decomposed, the hardware native gateset, 
number of qubits available to use as ancillas,
the topology (i.e. connectedness graph of the machine), and
the specifics of the gateset (e.g. on NISQ machines the number of CX gates could be an optimization goal, while
on fault-tolerant architectures the number of T gates would be a better optimization goal).
On one end of the spectrum, there exist numerous software packages (e.g. qiskit, tket \cite{Sivarajah_2020}) capable of performing such 
decompositions in a very general setting, using heurestic optimizations to minimize the number of gates.
On the other end of spectrum, and more in line with the topic of this article, researchers have studied singled 
out often recurring quantum primitives and focused on optimizing those.
For example, it is known that a Toffoli gate requires 6 CNOT gates, and this number is optimal under certain assumptions \cite{Shende_2008}.
The Toffoli gate also requires at least 4 T gates \cite{Howard_2017}. 
More generally, multiple-controlled Toffoli gate was also analyzed by various authors (\cite{Maslov_2016}, add others0.
Another well-studied example is the circuit for addition of two n-bit numbers \cite{Gidney2018halvingcostof, Perez_2017, Vedral_1996}.

Various quantum algorithms, such as Grover search or QAOA, require implementation of (arbitrary) diagonal operators. Inevitably, to execute those operators on the physical devices, they must be first decomposed into the hardware native gateset. The decomposition strategy may be greatly impacted by the qubit connectivity of the target device. The problem has recieved some attention and had previously been studied for a fully connected topology and a linear topology to some extent. \cite{10.5555/2011670.2011675,Shuch_2003, Welch_2014} provide vital results regarding theoretical lower bounds and describing general method of solving the problem in such case.
In this paper we plan to further extend them by designing a method of constructing generic diagonal circuits with lower connectivity of qubits, which may prove to be of more use in real-life situations.
In case of \cite{10.5555/2011670.2011675}, we also greatly extend the class of decomposable gates to diagonal gates with many more degrees of freedom.

In the following section we give a brief intuition about the structure of a solution in general case.
Later on we introduce three different variants of the problem (independent from target topologies) and describe the motivation behind them --- variants called $\WPA$ and $\SPA$ and their applications in quantum algorithms are a novelty and to the best of our knowledge have not been studied before by anyone.
Then we present comprehensive description of a general solution to decomposition problem on fully connected topology.
Then we describe the main goal of the paper, which is the solution to the problem on more restricted topologies --- i.e., a method for generating circuits having required properties.
Finally we present methods of exploiting symmetries contained in a particular diagonal operators to obtain lower $\CX$ count to the solutions in the general framework.

\subsection{Problem description}
Our primary goal is to devise a generic circuit $Q(\hat\theta)$ (for $\hat\theta=(\theta_1,\ldots,\theta_p)\in\mathbb{R}^p$) on $n$ qubits consisting solely of $\CX$ gates and exactly $p$ phase gates\footnote{In this paper, a phase gate is defined as $\Phase(\theta) = \operatorname{diag}(1, \, \exp(i \theta))$.} (parametrized by angles $\theta_1, \ldots, \theta_p$ respectively) such that, by adjusting $\hat\theta$, we can implement arbitrary diagonal operator (we will call $Q$ a \emph{generic diagonal circuit})\footnote{$Q$ being a generic circuit means here that the placement of all gates is already fixed; the only thing we can manipulate are angles parametrizing phase gates.}.
That is, given a diagonal operator on $n$ qubits (i.e., a diagonal matrix of dimension $2^n \times 2^n$) of a form $O = \sum_{k=0}^{2^n - 1} \exp(\alpha_k \cdot i)\ket{k} \bra{k}$ we can compute angles $\hat\theta = \hat\theta(\alpha_0, \ldots, \alpha_{2^n-1})$ such that $Q(\hat\theta)$ realizes $O$ up to a global phase.

From the last condition we can assume without a loss of generality that $\alpha_0 = 0$, since we can always remove $O$'s global phase by multiplying the matrix by $\exp(-\alpha_0\cdot i)$.
Therefore our problem has $2^n - 1$ real-valued degrees of freedom --- thus any valid $Q$ must have at least $2^n-1$ phase gates, as every phase gate introduces only one real parameter.
\section{General case}

Let us consider how circuits containing only phase and $\CX$ gates behave on computational basis states, which are $n$-bit long bitmasks.
Obviously, the phase gate does not change the aforementioned bitmask and $\CX$ gate changes it according to the rule presented in the \Cref{fig:cx-gate}.

\begin{figure}[H]
\begin{quantikz}[row sep={0.6cm,between origins}, column sep=0.3cm]
\lstick{$a$} & \ctrl{1} & \qw & \rstick{$a$} \\
\lstick{$b$} & \targ{}  & \qw & \rstick{$a \oplus b$} 
\end{quantikz}
\caption{Application of the $\CX$ gate is equivalent to changing the state of the target wire into the XOR (i.e., sum modulo 2) of the states before the gate application.}
\label{fig:cx-gate}
\end{figure}
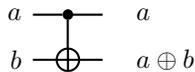

To analyze the behavior of phase gates in the circuit, it will be convenient to assign to every wire $k\in\{0,\dots,n-1\}$ at every point in time $t\in\mathbb{N}$, a nonempty subset of qubits such that its state is the sum of the corresponding bits from the original bitmask modulo 2.
\begin{definition}[Wire signature]
    For every wire $k\in\{0,\ldots,n-1\}$ let the \emph{signature of wire $k$ at time $t\in\mathbb{N}$} be a vector $v_k(t)\in\mathbb{F}_2^n$ such that:
    \[
        v_k(t) = \begin{cases}
            e_k\text{,} & \text{if $t=0$,}\\
            v_k(t-1) + v_l(t-1)\text{,} & \text{if $t>0$ and there is a gate $\CX(l, k)$ at time $t-1$,}\\
            v_k(t-1)\text{,} & \text{otherwise,}
        \end{cases}
    \]
    where $e_k$ denotes $k$-th vector from the standard basis of $\mathbb{F}_2^n$.
\end{definition}
Note that for initial state from the computational basis, state of a qubit $k$ at time $t$ is the inner product of $v_k(t)$ and the bitmask representing initial state (after performing a natural embedding into $\mathbb{F}_2^n$).

\begin{lemma}\label{lem:vkt-upper-bound}
    The number of different wire signatures achieved in a circuit $Q$, i.e. $\card{\{v_k(t): 0\leq k\leq n-1, 0\leq t\leq\len(Q)\}}$ is at most $n$ plus the number of $\CX$ gates in circuit $Q$.
\end{lemma}

\begin{proof}
Recall that all elements of the standard basis are introduced on the corresponding wires at time 0, before any $\CX$ is executed.
Any other signature must be achieved after execution of some $\CX$ gate.
It is clear that each $\CX$ introduces at most one new wire signature as the only wire changing its signature is the one corresponding to the target qubit.
%    Let $S = \{v_k(t): 0\leq k\leq n-1, 0\leq t\leq\len(Q)\}$ be the considered set of all achieved signatures.
%    Note that $S$ contains all $n$ elements $e_k$ of the standard basis, $0\leq k\leq n-1$, as they are introduced at time $0$.
%    For every $v\in S$ not from a standard basis, let $t_v$ be the first time it is seen in a circuit, i.e. the smallest $t$ such that there exists some $0\leq k\leq n-1$ such that $v_k(t)=v$.
%    Since $v\neq e_i$ for all wires $i$, we must have $t_v>0$ and $v_k(t_v-1)\neq v$, so $v$ must have been introduced by a $\CX$ gate with target wire $k$ at time $t_v$.
%    Therefore for every $v\in S$ not from a standard basis we can assign to it a unique controlled X gate introducing it --- therefore there are at least $\card{S}-n$ $\CX$ gates in $Q$.
\end{proof}

%\begin{figure}[h!]
%\centering
%\begin{quantikz}[row sep={0.6cm,between origins}, column sep=0.3cm, align equals at=1.5]
%& \ctrl{1} & \gate{P(\theta)} & \qw \\
%& \targ{}  & \qw              & \qw 
%\end{quantikz}
%$\;$=\begin{quantikz}[row sep={0.6cm,between origins}, column sep=0.3cm, align equals at=1.5]
%& \gate{P(\theta)} & \ctrl{1} & \qw \\
%& \qw              & \targ{}  & \qw 
%\end{quantikz}
%\caption{Elementary identity: phase gate commutes with control of a $\CX$ gate.}
%\label{fig:gate-identities}
%\end{figure}

\begin{proposition}\label{thm:simplify-phase-gates}
%Based on the identities presented in \Cref{fig:gate-identities} we conclude that
For any circuit consisting only of phase and $\CX$ gates there exists an equivalent circuit built using the same gateset and having the same number of $\CX$ gates such that for every non-zero $v\in\mathbb{F}_2^n$ there is at most one application of phase gate to some qubit $k$ at time $t$ such that $v_k(t)=v$.
%\begin{enumerate}
    % redundantne
    %\item does not contain two phase gates applied directly one after the other (on the same wire),
%    \item has the same number of $\CX$ gates,
%    \item does not contain a phase gate directly after some $\CX$'s control,
%    \item for every non-zero $v\in\mathbb{F}_2^n$ there is at most one application of phase gate on some qubit $k$ at time $t$ such that $v_k(t)=v$.
%\end{enumerate}
\end{proposition}

\begin{proof}
We will construct such circuit by modifying the original one.
%To satisfy (2), one simply needs to repeatedly apply identity shown on \Cref{fig:gate-identities}.
Note that, due to the choice of gateset, for any input computational basis state $\ket{b}$ the final state is of form $\exp(\alpha(b)\cdot i)\ket{f(b)}$, where $f$ is an isomorphism depending only on the placement of $\CX$es and $\alpha(b)$ is the sum of $\theta$s over all applications of gates $\Phase(\theta)$ to qubit $k$ at time $t$ such that $\sprod{v_k(t), b} = 1$.
Therefore, if there are two different applications of phase gates $\Phase(\theta_1), \Phase(\theta_2)$ to qubits $k_1,k_2$ at times $t_1,t_2$ (respectively) such that $v_{k_1}(t_1) = v_{k_2}(t_2)$, then they both contribute to the phase shift of the same subset of basis states --- thus we can remove the second phase gate and replace the first one with $\Phase(\theta_1 + \theta_2)$, keeping all final phase shifts intact.
This process decreases the total number of phase gates, so by repeating it we will eventually end up with an equivalent circuit satisfying the condition.
It remains to notice that the described modification does not change the number (as well as relative placement) of $\CX$ gates.
\end{proof}

\begin{corollary}\label{cor:phase-gate-count}
    Any circuit constructed in \Cref{thm:simplify-phase-gates} uses at most $2^n-1$ phase gates.
    In particular, applying \Cref{thm:simplify-phase-gates} to a generic diagonal circuit yields a generic diagonal circuit with exactly $2^n-1$ phase gates.
\end{corollary}

\begin{corollary}\label{cor:cx-lower-bound}
    The number of $\CX$ gates necessary to implement a generic diagonal circuit on $n$ qubits is the same as the minimal number of $\CX$ gates necessary to enumerate $\mathbb{F}_2^n\setminus\{0\}$ via $v_k(t)$, which is at least $2^n - n - 1$.
\end{corollary}

\begin{proof}
    Let $Q$ be any generic diagonal circuit and $Q'$ be the equivalent circuit obtained by applying \Cref{thm:simplify-phase-gates}.
    Let $c$ be the number of $\CX$ gates in $Q'$ (as well as in $Q$).
    From \Cref{cor:phase-gate-count} we know that $Q'$ uses exactly $2^n-1$ phase gates --- therefore $Q'$ must achieve all possible (non-zero) wire signatures, and from \Cref{lem:vkt-upper-bound} we have $c\geq 2^n-n-1$.
\end{proof}

\begin{figure}[!h]
\begin{quantikz}[row sep={0.6cm,between origins}, column sep=0.3cm]
\lstick{$v_0(0) = e_0$} & \ctrl{1} & \qw       & \targ{}   & \ctrl{1}  & \qw & \rstick{$v_0(4) =e_0 + e_2 + e_3$} \\
\lstick{$v_1(0) = e_1$} & \targ{}  & \qw       & \qw       & \targ{}   & \qw & \rstick{$v_1(4) =e_1 + e_2 + e_3$} \\
\lstick{$v_2(0) = e_2$} & \qw      & \targ{}   & \ctrl{-2} & \qw       & \qw & \rstick{$v_2(4) =e_2 + e_3$} \\
\lstick{$v_3(0) = e_3$} & \qw      & \ctrl{-1} & \qw       & \qw       & \qw & \rstick{$v_3(4) = e_3$}
\end{quantikz}
\caption{State of $k$-th qubit after execution of a $\CX$ circuit $Q$ is an inner product of original state's bitmask and $v_k(\len(Q))$. In this example, state $\ket{b} = \ket{1011}$ will be changed into $\ket{\sprod{v_0(4), b}\sprod{v_1(4), b}\sprod{v_2(4), b}\sprod{v_3(4), b}} = \ket{1001}$, and state $\ket{0010}$ will be changed into $\ket{1110}$.}
\end{figure}
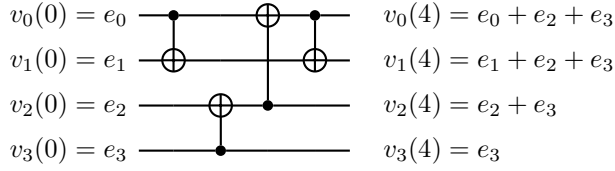

In the remainder of this paper we will focus on the problem of achieving \emph{shortest} generic diagonal circuit.
From \Cref{cor:phase-gate-count} we already know that the number of phase gates cannot be further optimized; our only concern left is making sure that every $v\in\mathbb{F}_2^n\setminus\{0\}$ will appear (as a wire signature) at some point in the circuit.
Based on this conclusion, we can reformulate the problem as: \emph{devise the shortest $\CX$ circuit on $n$ wires that enumerates $\mathbb{F}_2^n\setminus\{0\}$ via wire signatures}.\footnote{For a method of computing $\hat\theta$, refer to \Cref{sec:hat-theta}.}
Therefore, since from now on we will only consider circuits consisting of $\CX$ gates, $\len(Q)$ will naturally denote the number of $\CX$ gates in a circuit $Q$ --- in patricular we assume that at every time $t=0,\ldots,\len(Q)-1$ there will be exactly one $\CX$ gate applied.

Let us state an important property of all wire signatures at a fixed time $t$:
\begin{lemma}\label{lem:vkt-linear-independence}
    For any generic diagonal circuit $Q$ and any time $t$, vectors $v_k(t)$ ($k=0,\ldots,n-1$) form a basis of $\mathbb{F}_2^n$.
    %In particular, they are linearly independent.
\end{lemma}
\begin{proof}
    Induction on $t$.
    For $t=0$ all wire signatures form a standard basis.
    Assume that at time $t$ we apply a $\CX$ gate controlled by qubit $a$ and targeting qubit $b$.
    Therefore we have $v_b(t+1) = v_b(t) + v_a(t)$ and $v_k(t+1)=v_k(t)$ for all $k\neq b$.
    Adding one element of a linearly independent set to another does not affect linear independence --- and since the size of the set does not change, we obtain that wire signatures at time $t+1$ also form a basis of $\mathbb{F}_2^n$.
\end{proof}

Finally let us introduce a shorthand for vectors from $\mathbb{F}_2^n$: we will write $i_{n-1}i_{n-2}\ldots i_0$ for $[i_{n-1}, \ldots, i_0]\in\mathbb{F}_2^n$, e.g. $011$ will be the same as vector $e_0+e_1=[0,1,1]\in\mathbb{F}_2^3$.

\section{Problem variants}
\subsection{Exact decomposition}
This variant is also called \textbf{NPA (No Permutation Allowed)}. It requires that for every wire $k$ we have $v_k(\len(Q)) = v_k(0) = e_k$, i.e. the final $v_k$ is the same as the initial one. By solving this variant, we are able to decompose every diagonal operator exactly.

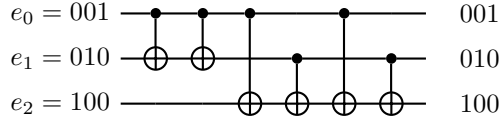
\begin{figure}[H]
\begin{quantikz}[row sep={0.6cm,between origins}, column sep=0.3cm]
\lstick{$e_0 = 001$} & \ctrl{1} & \ctrl{1} & \ctrl{2}  & \qw   & \ctrl{2}  & \qw  & \qw & \rstick{$001$} \\
\lstick{$e_1 = 010$} & \targ{}  & \targ{}  & \qw       & \ctrl{1}       & \qw  & \ctrl{1} &\qw & \rstick{$010$} \\
\lstick{$e_2 = 100$} & \qw      & \qw      & \targ{}   & \targ{}   & \targ{}  & \targ{} & \qw &\rstick{$100$} 
\end{quantikz}
\caption{One of the shortest circuits for $n = 3$ for NPA variant. The sequence of signatures generated on target wires is $[011, 010, 101, 111, 110, 100]$ --- each $v\in\mathbb{F}_2^n\setminus\{0\}$ was generated (with $001$ available before any $\CX$ gate) and the wire signatures at the end of execution were $[001, 010, 100]$.}
\end{figure}

\subsection{WPA - Wire Permutation Allowed}
This variant requires that $v_k(\len(Q)) = e_{\sigma(k)}$ for some permutation of wires $\sigma$. Circuits of this form may perform the diagonal operator exactly, changing only the logical layout of qubits, which may slightly influence the layout of gates following the diagonal operator (this influence may as well positively affect the overall gate count).

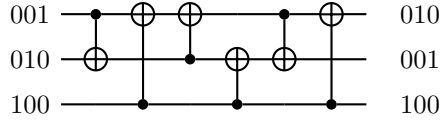
\begin{figure}[H]
\begin{quantikz}[row sep={0.6cm,between origins}, column sep=0.3cm]
\lstick{$001$} & \ctrl{1} & \targ{}  & \targ{}  & \qw      & \ctrl{1}  & \targ{}   & \qw & \rstick{$010$} \\
\lstick{$010$} & \targ{}  & \qw      & \ctrl{-1}& \targ{}  & \targ{}   & \qw       & \qw & \rstick{$001$} \\
\lstick{$100$} & \qw      & \ctrl{-2}& \qw      & \ctrl{-1}& \qw       & \ctrl{-2} & \qw & \rstick{$100$} 
\end{quantikz}
\caption{One of the shortest circuits for $n = 3$ for WPA variant. The sequence of signatures generated on target wires is $[011, 101, 110, 111, 001, 010]$ --- each $v\in\mathbb{F}_2^n\setminus\{0\}$ was generated (with $100$ available before any $\CX$ gate) and the wire signatures at the end of execution were $[010, 001, 100]$. The shortest solution to WPA and NPA on fully-connected topology have the same length, see \Cref{cor:wpa-npa-fully-conn}.}
\end{figure}

\subsection{SPA - State Permutation Allowed}

This variant requires only that every $v\in\mathbb{F}_2^n\setminus\{0\}$ is enumerated. In this variant, the diagonal operator applies the phases correctly but the final wire signatures, i.e. $\{v_k(\len(Q)): 0\leq k\leq n-1\}$ may represent different basis of $\mathbb{F}_2^n$ from the initial one.

\begin{figure}[H]
\begin{quantikz}[row sep={0.6cm,between origins}, column sep=0.3cm]
\lstick{$001$} & \ctrl{1} & \qw     & \ctrl{2} & \qw       & \qw & \rstick{$001$} \\
\lstick{$010$} & \targ{}  & \ctrl{1}& \qw      & \targ{}   & \qw & \rstick{$101$} \\
\lstick{$100$} & \qw      & \targ{} & \targ{}  & \ctrl{-1} & \qw & \rstick{$110$} 
\end{quantikz}
\caption{One of the shortest circuits for $n = 3$ for SPA variant. The sequence of signatures generated on target wires is $[011, 111, 110, 101]$ - each nontrivial $v\in\mathbb{F}_2^n\setminus\{0\}$ was enumerated and the wire signatures at the end of execution were $[001, 101, 110]$.}
\end{figure}
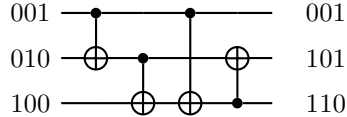

This variant may be useful in cases where we want to compute $O \cdot U \cdot O^\dagger$, where $U$ is some unitary, agnostic to the order of states. Examples include Grover's algorithm for which the oracle is a diagonal operator satisfying $O = O^\dagger$ and $U=G_n$ is Grover's diffusion operator. Thence, the algorithm may as well be implemented as a sequence of interchanging $O$ and $O^\dagger$ separated by Grover diffusion operators, which are state ordering agnostic.

%\todo{verify experimentally że to ma sens}
% i tried on n = 3 and it worked XD - Wielomian
\begin{figure}[H]
\begin{quantikz}[row sep={0.6cm,between origins}, column sep=0.3cm, align equals at=1]
 & \gate{\NPA_n}\qwbundle[alternate]{} & \gate{G_n}\qwbundle[alternate]{} & \gate{\NPA_n}\qwbundle[alternate]{} & \gate{G_n}\qwbundle[alternate]{} & \qwbundle[alternate]{}
\end{quantikz}
\;=%
\begin{quantikz}[row sep={0.6cm,between origins}, column sep=0.3cm, align equals at=1]
 & \gate{\SPA_n}\qwbundle[alternate]{} & \gate{G_n}\qwbundle[alternate]{} & \gate{\SPA_n^\dagger}\qwbundle[alternate]{} & \gate{G_n}\qwbundle[alternate]{} & \qwbundle[alternate]{}
\end{quantikz}
\caption{Two calls to the oracle in Grover's algorithm can be implemented equivalently as SPAs, reducing the total gate count.}
\end{figure}

From a simple observation that every solution to NPA is also a solution to WPA, and every solution to WPA is also a solution to SPA we get the following
\begin{corollary}\label{cor:variants-inequalities}
    If $\NPA(n)$,  $\WPA(n)$ and $\SPA(n)$ denote the length of shortest solutions for variants NPA, WPA, and SPA on $n$ qubits, respectively, then
    \[
        \SPA(n) \leq \WPA(n) \leq \NPA(n)\text{.}
    \]
\end{corollary}

\section{Solution}

Solution to any variant is strongly dependent on the topology of the device. We describe solutions to a selection of topologies.

\subsection{Fully connected topology}
The first topology considered in this paper will be the fully connected topology, i.e. there exists a connection between any wires $a$ and $b$, $a, b \in \{0, 1, \ldots, n - 1\}$. Analysis of the variants in this case provides insight into the problem, allowing experimentation with enumeration of the wire signatures.

\begin{theorem}\label{thm:wpa-lower-bound}
On the fully-connected topology, any solution to WPA problem has length at least $2^n - 2$. 
\end{theorem}

% \todo{skroc te dowody. np pierwszy akapit by lemma 6 there is no signature ea + eb <- od razu wynika. Nie pisac at some point t itd. Drugi akapit tez troche dlugi ale na razie nie ma pomyslu jak to tam.}
\begin{proof}
Let $S$ be a solution to the WPA problem. Consider the set of wires $T$ that are a target to at least one $\CX$ gate in $S$. Suppose that $\card{T} \leq n - 2$. Then there exist wires $a, b \not \in T$, $a \neq b$.
% Consider the time at which $e_a + e_b \in \mathbb{F}_2^n\setminus\{0\}$ has been enumerated. It must have been after some $\CX$ target, since $e_a+e_b$ does not belong to the standard basis. This $\CX$ target could not be on neither $a$-th nor $b$-th wire since $a, b \not \in T$. Thus, at some point $t$ and for some wire $c\not\in\{a,b\}$, we have $v_a(t) = e_a$, $v_b(t) = e_b$, and $v_c(t) = e_a + e_b$ --- this is impossible, since those signatures are linearly dependent, which contradicts \Cref{lem:vkt-linear-independence}.
Note that $e_a + e_b \in \mathbb{F}_2^n\setminus\{0\}$ couldn't have been enumerated: otherwise, for some wire $c$ and at some time $t$ we would have $v_a(t) = e_a$, $v_b(t) = e_b$, and $v_c(t) = e_a + e_b$, which contradicts \Cref{lem:vkt-linear-independence}.

Thus, $\card{T} \geq n - 1$. For each wire $k\in T$, consider the last target on this wire. It must have left the wire with signature $e_{\sigma(k)}$ --- this is the WPA constraint.
Let us define $B=\{e_{\sigma(k)}: k\in T\}$ and $B'=\mathbb{F}_2^n\setminus(\{e_i: 0\leq i\leq n-1\}\cup\{0\})$.
Note that $B$ and $B'$ are disjoint sets.
Moreover, for every signature $v\in B\cup B'$ there is at least one $\CX$ gate in $S$ introducing that signature at its target wire --- for $e_{\sigma(k)}\in B$ it is the last CX targetting wire $k$, and for $v\in B'$ this is true because $S$ is a correct solution and $v$ is not from a standard basis.
Therefore $\len(S)$, which is the total number of $\CX$ gates in $S$, satisfies $\len(S) \geq \card{B\cup B'} = \card{B} + \card{B'} = \card{T} + 2^n - n - 1 \geq 2^n - 2$.
\end{proof}

\begin{theorem}\label{thm:npa-fully-connected}
On the fully-connected topology, the shortest solution to NPA problem has length $2^n - 2$. 
\end{theorem}

\begin{proof}
Since any solution to NPA is also a solution to WPA, by \Cref{thm:wpa-lower-bound} the shortest solution to NPA must have length at least $2^n - 2$.

We will now construct the solution of length $2^n - 2$ recursively. For $n = 1$ we do not place any $\CX$, and obtain solution of cost $ 0 = 2^1 - 2$. To construct solution for $n + 1$ wires, we apply the previously created solution to the first $n$ wires. Then, we enumerate $\{v\in \mathbb{F}_2^{n+1}: \sprod{e_n, v} = 1\}$ using Gray codes\footnote{Or any equivalent algorithm that is able to enumerate all bitmasks (starting and ending with 0 bitmask) by performing $2^n$ bitflips.}. To do that, we place $2^n$ additional $\CX$es in such a way, that the control qubit is placed in the corresponding Gray code bit and the target is set to the last wire.

Thus, the recursive solution to $NPA_{n+1}$ contains $(2^n - 2) + 2^n = 2 \cdot 2^n - 2 = 2^{n + 1} - 2$ $\CX$ gates. 
\end{proof}

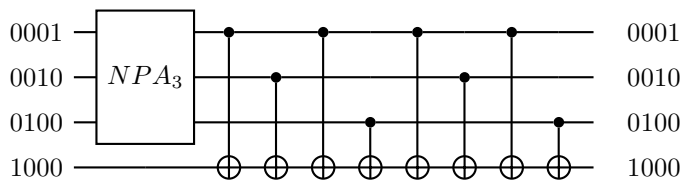
\begin{figure}[H]
\begin{quantikz}[row sep={0.6cm,between origins}, column sep=0.3cm]
\lstick{$0001$} & \gate[wires = 3]{NPA_3} & \ctrl{3} & \qw      & \ctrl{3} & \qw      & \ctrl{3} & \qw      & \ctrl{3} & \qw      & \qw & \rstick{$0001$} \\
\lstick{$0010$} &                         & \qw      & \ctrl{2} & \qw      & \qw      & \qw      & \ctrl{2} & \qw      & \qw      & \qw & \rstick{$0010$} \\
\lstick{$0100$} &                         & \qw      & \qw      & \qw      & \ctrl{1} & \qw      & \qw      & \qw      & \ctrl{1} & \qw & \rstick{$0100$} \\
\lstick{$1000$} & \qw                     & \targ{}  & \targ{}  & \targ{}  & \targ{}  & \targ{}  & \targ{}  & \targ{}  & \targ{}  & \qw & \rstick{$1000$} 
\end{quantikz}
\caption{Example solution for NPA for $n = 4$ wires. After recursive call, $\{v\in \mathbb{F}_2^{n+1}: \sprod{e_n, v} = 1\}$ is enumerated in the order of Gray codes: $[1001, 1011, 1010, 1110, 1111, 1101, 1100, 1000]$. All remaining signatures are enumerated by recursion.}
\end{figure}

\begin{corollary}\label{cor:wpa-npa-fully-conn}
    Assuming fully connected topology, the shortest solutions to WPA are of the same length as the shortest solutions to NPA, that is $2^n-2$.
\end{corollary}
\begin{proof}
    From \Cref{thm:wpa-lower-bound}, \Cref{cor:variants-inequalities} and \Cref{thm:npa-fully-connected} we get
    \[
        2^n-2 \leq \WPA(n) \leq \NPA(n) = 2^n-2\text{.}
    \]
\end{proof}

Additional conclusion from the above is that there is no use for WPA solutions for the fully connected architectures - the required $\CX$ count is the same as in NPA and there is no benefit to different logical order of wires after the operator implementation, since any $\CX$ following the diagonal operator can be applied to any logical wire ordering.

\begin{theorem}
On the fully-connected topology, the shortest solution to SPA problem has length $2^n - n - 1$.
\end{theorem}

\begin{proof}
We see that the lower bound is $2^n - n - 1$ by \Cref{cor:cx-lower-bound}.

To prove that the lower bound may in fact be obtained we provide a recursive construction similar to the NPA solution. For $n = 1$ we provide empty circuit, having $ 0 = 2^1 - 1 - 1$ $\CX$ gates. To recursively construct circuit $SPA_{n+1}$, we start with Gray codes enumeration of $\{v\in \mathbb{F}_2^{n+1}: \sprod{e_n, v} = 1\}$ except from the last $\CX$ (thus we introduce $2^n - 1$ $\CX$ gates). We then apply $SPA_n$ to the first $n$ qubits. All signatures $v$ such that $\sprod{v, e_n}=1$ have been enumerated during the Gray code phase, and all signatures $v$ such that $\sprod{v, e_n}=0$ are generated by the recursion. 

The total number of $\CX$ operations is $(2^n - 1) + (2^n - n - 1) = 2^{n + 1} - (n + 1) - 1$, which proves the induction step.
\end{proof}

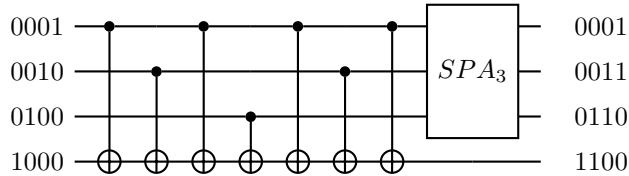
\begin{figure}[H]
\begin{quantikz}[row sep={0.6cm,between origins}, column sep=0.3cm]
\lstick{$0001$} & \ctrl{3} & \qw      & \ctrl{3} & \qw      & \ctrl{3} & \qw      & \ctrl{3} & \gate[wires = 3]{SPA_3} & \qw & \rstick{$0001$} \\
\lstick{$0010$} & \qw      & \ctrl{2} & \qw      & \qw      & \qw      & \ctrl{2} & \qw      &                         & \qw & \rstick{$0011$} \\
\lstick{$0100$} & \qw      & \qw      & \qw      & \ctrl{1} & \qw      & \qw      & \qw      &                         & \qw & \rstick{$0110$} \\
\lstick{$1000$} & \targ{}  & \targ{}  & \targ{}  & \targ{}  & \targ{}  & \targ{}  & \targ{}  & \qw                     & \qw & \rstick{$1100$} 
\end{quantikz}
\caption{Example solution for SPA for $n = 4$ wires. During the first phase, $\{v\in \mathbb{F}_2^{n+1}: \sprod{e_n, v} = 1\}$ is enumerated in the order of Gray codes: $[1001, 1011, 1010, 1110, 1111, 1101, 1100]$. All remaining signatures are enumerated by recursion.}
\end{figure}

\subsection {Linear topology}

The linear topology, which consists of connections between qubits $a$ and $b$ if and only if $\abs{a-b}=1$, is an interesting case of the diagonal decomposition problem. On the one hand, some devices have this exact topology, as the qubits are physically located on a line. On the other hand, for devices with slighly richer, but still sparse topologies, the line may often be embedded as a subgraph. This versatility of the line topology makes it useful on a large collection of graphs.

\begin{definition}[$V$ circuit]
Let $i$, $j$ be two different wires.
We define circuit $V(i,j)$ recursively in $|i - j|$ as
\[
    V(i,j) \coloneqq
    \begin{cases}
        \CX(i,j)\text{,} & \text{for } \abs{i-j} = 1\text{,}\\
        \CX(i,i+d(i,j))\circ V(i+d(i,j), j)\circ\CX(i,i+d(i,j))\text{,} & \text{otherwise,}
    \end{cases}
\]
where $d(i, j) = \sgn(j - i)$.
\end{definition}
It is worth noting that $V(i,j)$ operates only on wires between $i$ and $j$, and consists of exactly $2 \cdot \abs{i - j} - 1$ $\CX$ gates.

\begin{figure}[H]
\begin{quantikz}[row sep={0.6cm,between origins}, column sep=0.3cm]
\lstick{$00001$} & \ctrl{1} & \qw      & \qw      & \qw      & \qw      & \qw      & \ctrl{1} & \qw      & \rstick{$00001$} \\
\lstick{$00010$} & \targ{}  & \ctrl{1} & \qw      & \qw      & \qw      & \ctrl{1} & \targ{}  & \qw      & \rstick{$00010$} \\
\lstick{$00100$} & \qw      & \targ{}  & \ctrl{1} & \qw      & \ctrl{1} & \targ{}  & \qw      & \qw      & \rstick{$00100$} \\
\lstick{$01000$} & \qw      & \qw      & \targ{}  & \ctrl{1} & \targ{}  & \qw      & \qw      & \qw      & \rstick{$01000$} \\
\lstick{$10000$} & \qw      & \qw      & \qw      & \targ{}  & \qw      & \qw      & \qw      & \qw      & \rstick{$11111$} 
\end{quantikz}
\caption{Example of the $V(0, 4)$ circuit on five wires. The final signature ($v_k(7)$) differs only in the last qubit, which is now a sum of all $v_k(0)$ for wires $k=0,\ldots,4$. }
\end{figure}
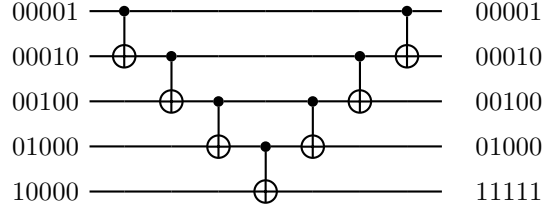

\begin{lemma}\label{lem:vij-circuit}
    Let $i \neq j$ be two wires, and let $S$ denote the set of all wires between $i$ and $j$ (inclusively).
    After applying circuit $V(i,j)$:
    \[
    v_k(\len(V(i,j))) =
    \begin{cases}
        \sum_{l\in S} v_l(0)\text{,} & \text{for $k=j$,}\\
        v_k(0)\text{,} & \text{otherwise.}
    \end{cases}
    \]
\end{lemma}
\begin{proof}
    By renumbering wires appropriately, we can assume that $0 = j < i$.
    We proceed by induction on $i$.
    The case of $i=1$ is trivial.
    By the recursive definition, $V(i+1, 0)=\CX(i+1,i)\circ V(i,0)\circ \CX(i+1,i)$.
    First applied gate, $\CX(i+1,i)$, sets $v_{i}(1) = v_{i}(0) + v_{i+1}(0)$.
    Then the application of $V(i, 0)$ sets the signature of wire $0$ to
    \[
        \sum_{k=0}^{i} v_k(1) = \sum_{k=0}^{i+1} v_k(0)\text{,}
    \]
    leaving all remaining signatures intact due to the inductive assumption.
    Finally we apply the last gate $\CX(i+1,i)$, which sets the signature of $i$ to $$v_i(2i + 1) = v_i(2i) + v_{i + 1}(2i) = v_{i}(1) + v_{i+1}(0) = v_i(0)\text{.}$$
\end{proof}

We recursively define a circuit $\GRAY_n$ acting on $n$ wires.
$\GRAY_1$ is an empty circuit. Construction of $\GRAY_{n+1}$ depends on the parity of $n+1$:
\begin{itemize}
    \item If $n+1$ is even, we place $\GRAY_n$ on the last $n$ wires, apply $V(0, (n+1)/2)$, and we finish with $\GRAY_n$ on the last $n$ wires.
    \item If $n+1$ is odd, we place $\GRAY_n$ on the first $n$ wires, apply $V(n, n/2)$, and we finish with $\GRAY_n$ on the first $n$ wires.
\end{itemize}

\begin{figure}[H]
\centering
\begin{subfigure}[b]{\textwidth}
\centering
\begin{quantikz}[row sep={0.6cm,between origins}, column sep=0.3cm]
\lstick{} & \qw            & [4mm] \qw                    & \gate[wires=3]{V(0, m)}  & \qw  & \qw \\
\lstick{} & \qwbundle{m-1} & \gate[wires=3]{\GRAY_{2m-1}} &  & \gate[wires=3]{\GRAY_{2m-1}} & \qw   \\
\lstick{} & \qw            &                              &  &  & \qw  \\
\lstick{} & \qwbundle{m-1} &                              & \qw  &  & \qw 
\end{quantikz}
\caption{ Recursion formula for $\GRAY_n$ for $n=2m$ being an even number of wires.  }
\end{subfigure}

\begin{subfigure}[b]{\textwidth}
\centering
\begin{quantikz}[row sep={0.6cm,between origins}, column sep=0.3cm]
\lstick{} & \qwbundle{m}   & [4mm] \gate[wires=3]{\GRAY_{2m}} & \qw  & \gate[wires=3]{\GRAY_{2m}}  & \qw \\
\lstick{} & \qw            &                           & \gate[wires=3]{V(2m, m)} &  & \qw   \\
\lstick{} & \qwbundle{m-1} &                           &   &  & \qw  \\
\lstick{} & \qw            & \qw                       &   & \qw & \qw 
\end{quantikz}
\caption{ Recursion formula for $\GRAY_n$ for $n=2m+1$ being an odd number of wires.  }
\end{subfigure}
\caption{$\GRAY_n$ circuit.}
\end{figure}
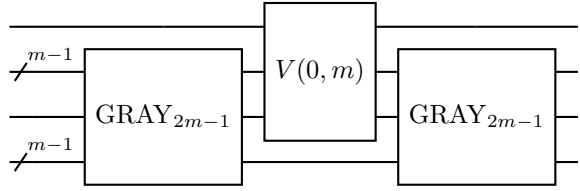
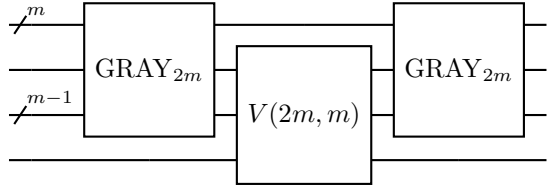

\begin{lemma}\label{lem:gray-circuit}
Let $n$ be the number of available wires and let $m=\floor{n/2}$ denote the \emph{middle wire}.
\begin{enumerate}
    \item Circuit $GRAY_n$ enumerates all signatures $v\in\mathbb{F}_2^n$ such that $\sprod{v, e_m}=1$ --- all of them on the wire $m$,
    \item After the execution of $\GRAY_n$, signatures of all wires other than $m$ are left unchanged.
        In particular, wire $m$ will be the only wire which signature $v$ satisfies $\sprod{v, e_m}=1$.
\end{enumerate}
\end{lemma}

\begin{proof}
Induction on $n$.
For $n = 1$, there is only one signature to enumerate (which is $e_0$) and it is already present on the wire $\floor{n/2}=0$, so we do not need any $\CX$ gates, which proves the base case.
Let us assume the lemma is correct for $1, 2, \ldots, n - 1$. We perform the induction step based on the parity of $n$.

Let $n = 2m$. By the recursive formula for $\GRAY_n$, we first execute $\GRAY_{n-1}$ on the last $n - 1$ wires. By the inductive assumption, this circuit enumerates on the $(\floor{(n-1)/2}+1)=m$-th wire all signatures $v$ satisfying $\sprod{v, e_m} = 1$ and $\sprod{v, e_0}=0$ (as the circuit does not act on wire $0$), where ``$+ 1$'' part comes from the fact that the circuit starts at wire numbered with 1.
By \Cref{lem:vij-circuit} and from the inductive assumption, after execution of $V(0, m)$ the signature of the wire $m$ will be some $v_m(t)\in\mathbb{F}_2^n$ such that $\sprod{v_m(t), e_0}=\sprod{v_m(t), e_m}=1$, and for every wire $k\in\{1,\ldots,n-1\}\setminus\{m\}$ we have $\sprod{v_k(t), e_0}=\sprod{v_k(t), e_m}=0$.
Therefore, during execution of the final $\GRAY_{n-1}$, we enumerate (again on the $m$-th wire) all signatures $v$ such that $\sprod{v, e_0}=\sprod{v,e_m}=1$.

Now let $n = 2m + 1$. Similarly, the $\GRAY_{n-1}$ applied to the first $n - 1$ wires enumerates on $m$-th wire every signature $v$ satisfying $\sprod{v, e_m}=1$ and $\sprod{v, e_{n-1}}=0$.
By applying $V(n-1,m)$ followed by $\GRAY_{n-1}$, we enumerate all signatures $v$ such that $\sprod{v, e_m} = \sprod{v, e_{n-1}}=1$.

Thus, in both cases of parity of $n$, we enumerate on the wire $m$ all signatures $v$ satisfying $\sprod{v, e_m}=1$, which proves (1).
To show (2), it suffices to notice that all applications of circuit $V$ (including those in recursive calls) affect only wire $m=\floor{n/2}$. \end{proof}

\begin{lemma}
The circuit $\GRAY_n$ consists of $5\cdot 2^n/6 - n + (-1)^n/6 - 1/2$ $\CX$ gates.
\end{lemma}

%\todo{include lemma that $\GRAY$ finishes at singleton subsets on controlling wires?}

\begin{proof}
The proof is a straightforward induction on $n$, with a casework on the parity of $n$. In both cases, $\GRAY_n$ consists of two $\GRAY_{n-1}$ circuits and a single $V$ circuit -- in case of even $n$, the appropriate circuit $V$ is $V(0, n / 2)$ and in case of odd $n$: $V(n - 1, (n - 1) / 2)$.
%For $n = 1$ the $\GRAY$ circuit is empty and has 0 $\CX$es. We see that $0 = \frac{(5 \cdot 2 - 1 - 6 - 3)}{6}$.

%\todo{remove explicit computation}
%Let us assume the formula is correct for $1, 2, \ldots, n - 1$. We have two cases:
%\begin{itemize}
%    \item $n=2m$ is even. Then $\GRAY_n$ consists of two $\GRAY_{n-1}$ and a single $V(0, m)$. Thus, the number of $\CX$es is $2 \cdot \frac{(5 \cdot 2^{n-1} - 1 - 6(n - 1) - 3)}{6} + 2 \cdot \frac{n}{2} - 1 = \frac{(5 \cdot 2^n - 12 n + 4 + 6 n - 6)}{6} = \frac{(5 \cdot 2^n + (-1)^n - 6n - 3)}{6}$ since $n$ is even.
%    \item $n=2m+1$ is odd. Then $\GRAY_n$ consists of two $\GRAY_{n-1}$ and a single $V(2m, m)$. Thus, the number of $\CX$es is $2 \cdot \frac{(5 \cdot 2^{n-1} + 1 - 6(n - 1) - 3)}{6} + 2 \cdot \frac{n - 1}{2} - 1 = \frac{(5 \cdot 2^n - 12 n + 8 + 6 n - 12)}{6} = \frac{(5 \cdot 2^n + (-1)^n - 6n - 3)}{6}$ since $n$ is odd.
%\end{itemize}
\end{proof}

\begin{theorem}\label{thm:linear-spa-upper-bound}
% On linear topology, the shortest solution to SPA has at most $\frac{5}{3}\cdot 2^n + O(n^2)$ $\CX$ gates.
On linear topology, there exists a solution to SPA that uses $\frac{5}{3}\cdot 2^n + O(n^2)$ $\CX$ gates (explicitly described in the proof).
In particular, the shortest solution to SPA on linear topology has at most $\frac{5}{3}\cdot 2^n + O(n^2)$ $\CX$ gates.
\end{theorem}
%\todo{is the constant factor 5/3 or 5/3 + eps due to the positive polynomial residuum?}

\begin{proof}
We construct a solution to SPA of specified size recursively. For $n = 1$, we provide empty circuit. For $n > 1$, we construct $\SPA_n$ in three phases:
\begin{enumerate}
    \item $\GRAY_n$ on all wires.
    \item a sequence of $\SWAP$ gates: $\SWAP(\floor{n/2}, \floor{n/2} + 1)$, $\SWAP(\floor{n/2} + 1, \floor{n/2} + 2)$, $\ldots$, $\SWAP(n - 2, n - 1)$. This is equivalent to performing a cyclic rotation on the last $\ceil{n/2}$ wires\footnote{
    Since we only need to ``push`` $e_{m}$ from $m$-th qubit to the last one, we could replace every $\SWAP(i, i + 1)$ gate with $\CX(i, i + 1)$ followed by $\CX(i + 1, i)$.
    This optimization reduces total number of CX gates in step (2) by $1/3$rd.
    For clarity and since it does not change the asymptotic behavior, the proof of correctness is left to the interested reader.}.
    \item $\SPA_{n-1}$ on the first $n - 1$ wires. 
\end{enumerate}
% https://b.dev.beit.tech/b2/circuits/quirk/#t=CX;1CX;11CX;111CX;1111CX;XC;1XC;11XC;111XC;1111XC
%\todo{te swapy można zrobix w 1/3 mniej CXów (link w komentarzu), please consult JanT}
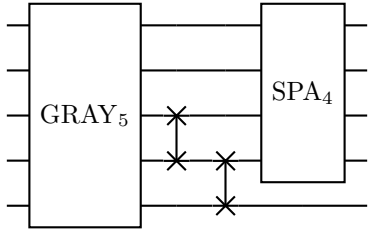
\begin{figure}[H]
\begin{quantikz}[row sep={0.6cm,between origins}, column sep=0.3cm]
\qw & \gate[wires=5]{\GRAY_5} & \qw      & \qw      & \gate[wires=4]{\SPA_4} & \qw \\
\qw &                         & \qw      & \qw      &                        & \qw \\
\qw &                         & \swap{1} & \qw      &                        & \qw \\
\qw &                         & \targX{} & \swap{1} &                        & \qw \\
\qw &                         & \qw      & \targX{} & \qw                    & \qw
\end{quantikz}
\caption{ Example of a $SPA_n$ circuit for $n=5$.  }
\end{figure}

Obviously, the $\GRAY_n$ circuit enumerates all signatures $v$ such that $\sprod{v, e_m}=1$.
%\todo{include theorem about not changing states in gray?}
Notice that due to \Cref{lem:gray-circuit} after performing step (2) the signatures of the first $n-1$ wires are exactly $e_0, \ldots, e_{m - 1}, e_{m + 1}, \ldots, e_{n-1}$.
Therefore by recursively performing $\SPA_{n-1}$ on the first $n-1$ wires, we will enumerate all signatures $v\neq 0$ such that $\sprod{v, e_m}=0$.

Thus the above circuit solves SPA problem.
Now we claim that it contains exactly\footnote{Using the above optimization of $\SWAP$s in the step (2), we arrive at the formula $5/3 \cdot 2^n -2n - ((-1)^n + 9) / 6$. } $$5 / 3 \cdot 2^n + (6n^2 -60n - 33 -7(-1)^n) / 24 \;\CX\text{ gates.}$$

The proof of the $\CX$ count formula is again a straightforward induction on $n$ with a casework depending on the parity of $n$. Each such circuit contains $\GRAY_n$, $\SPA_{n - 1}$ and an appropriate number of $\SWAP$ gates (each consisting of exactly 3 $\CX$ gates): $(n - 2) / 2$ in the even case, and $(n - 1) / 2$ in the odd case. 

%For $n = 1$, the empty circuit contains $0 = \frac{5 \cdot 2 + 2}{3} + 1 - 4 - 1$ $\CX$es. For $n > 1$ we observe that there are $2 \cdot (n - \floor{\frac{n}{2}} - 1)$ SWAP gates. We have two cases:
%\begin{itemize}
%    \item $n$ is even. Then number of SWAP gates is $n - 2$. Every SWAP consists of 3 $\CX$ gates. Thus, the total number of $\CX$ gates is $3n - 6 + (\frac{5 \cdot 2^n + 1 - 6n - 3}{6}) + (\frac{5 \cdot 2^{n-1} + 2}{3} + (n-1)^2 -4(n - 1) - 1) = \frac{5 \cdot 2^n -2}{3} -n + 1 + n^2 - 2n + 1 -4n + 4 - 1 + 3n - 6 = \frac{5 \cdot 2^n -2 \cdot (-1)^n}{3} + n^2 - 4n  - 1$.
%    \item $n$ is odd. Then number of SWAP gates is $n - 1$.  Thus, the total number of $\CX$ gates is $3n - 3 + (\frac{5 \cdot 2^n - 1 - 6n - 3}{6}) + (\frac{5 \cdot 2^{n-1} - 2}{3} + (n-1)^2 -4(n - 1) - 1) = \frac{5 \cdot 2^n +2}{3} -n - 2 + n^2 - 2n + 1 -4n + 4 - 1 + 3n - 3 = \frac{5 \cdot 2^n -2 \cdot (-1)^n}{3} + n^2 - 4n  - 1$.
%\end{itemize}

Thus, the described circuit solves SPA problem with $\frac{5}{3} \cdot 2^n + O(n^2)$ $\CX$ gates.

\end{proof}

\begin{theorem}
Let $n \geq 2$ be the number of wires. Assume a circuit consisting of $\CX$ gates satisfies the following criteria:
\begin{enumerate}[(i)]
    \item the middle wire (i.e., wire $m \coloneqq \floor{n/2}$) is never a control wire,
    \item all wire signatures $v$ satisfying $\sprod{v, e_m}=1$ appear on $m$-th wire.
\end{enumerate}
Then the circuit contains at least $\frac{(5 \cdot 2^n + (-1)^n - 6n - 3)}{6}$ gates.
In other words, the circuit $\GRAY_n$ is optimal under those assumptions.
\end{theorem}

\begin{proof}
By $\# \CX(i, j)$ we denote the number of $\CX(i, j)$ gates in the sequence.
Observe that there are $2^{n-1}$ states that need to appear on the $m$-th wire, and the state of $m$-th wire can only be altered by applying $\CX(m-1, m)$ or $\CX(m+1, m)$. Thus,
\begin{equation}
    \# \CX(m-1, m) + \# \CX(m+1, m) \geq 2^{n-1} - 1.
\end{equation}
Let us define
\begin{equation*}
    T \coloneqq \{ v \in \mathbb{F}_2^n : \sprod{v, e_m} = 1 \}.
\end{equation*}

Now fix $k > 1$. We will process the circuit one gate at a time.
%At any point of the time, let $w_0, w_1, \dots w_{n-1}$ be the current states of wires $0, 1, \dots n-1$.
For any point in time $t$, define
\begin{equation*} 
    A(t) := T \cap \textrm{span}(v_{m-k+1}(t), \dots v_{m+k-1}(t)).
\end{equation*}
For example, at the beginning of the circuit we have
\begin{equation*}
    A(0) = \{ v \in \mathbb{F}_2^n : \sprod{v, e_m} = 1 \land \sprod{v, e_l} = 0 \textrm{ if } \abs{l-m} \geq k \}\text{.}
\end{equation*}
Let us consider a $\CX$ gate applied at time $t$.
One can verify that:
\begin{itemize}
    \item $ \card{A(t+1)}  = 2^{2k - 2}$,
    \item If the gate is $\CX(m-k, m-k+1)$ or $\CX(m+k, m+k-1)$, then $\card{A(t+1) \cap A(t)} = 2^{2k-3}$ and $\card{A(t+1) \setminus A(t)} = 2^{2k-3}$,
    \item Otherwise, $A(t+1) = A(t)$.
\end{itemize}
Note that condition (ii) can be restated as: ``For every $v \in T$, there is a time $t$ when $v_m(t) = v$''. It follows that for every $v \in T$, there is a time $t$ when $v \in A(t)$. Thus
\begin{equation*}
    \card{T} \leq \card{A(0)} + 2^{2k-3} ( \# \CX(m-k, m-k+1) + \# \CX(m+k, m+k-1) )\text{,}
\end{equation*}
which simplifies to
\begin{equation}
    2^{n - 2k + 2} - 2 \leq \# \CX(m-k, m-k+1) + \# \CX(m+k, m+k-1).
\end{equation}

We are almost done. In case of $n$ odd, adding inequalities (1) and (2) for $k=2 \dots (n-1)/2$ completes the proof. In case of $n$ even, the $0$-th wire is left without a pair. By mimicking the previous argument, one can show that $\# \CX(0, 1) \geq 2$, and again adding all inequalities yields the desired lower bound. We leave the details to the interested reader.
\end{proof}

Let us consider a \emph{reachability problem}, in which our goal is to achieve given sequence of wire signatures on $n$ wires using the minimal number of $\CX$ gates, i.e. given a basis $x_0, \ldots, x_{n-1}$ of $\mathbb{F}_2^n$, we want to find shortest circuit $Q$ such that $v_k(\len(Q))=x_k$ for all wires $k$.

\begin{theorem}\label{thm:reachability-upper-bound}
On linear topology on $n$ wires, any valid configuration of wire signatures (i,e., any $v_k(t)$s forming a basis of $\mathbb{F}_2^n$) may be reached starting from a standard basis in $O(n^2)$ $\CX$ gates.
\end{theorem}

\begin{proof}
%We will write $i \in S(k)$ when we mean that $i$-th bit is set on the $k$-th wire. 
%Note, that the truth value of $i \in S(k)$ may change after applying some $\CX$ gate but the value will be clear from the context.
%
Let $x_0, x_1, \ldots, x_{n - 1}\in\mathbb{F}_2^n$ be wire signatures we want to reach. Let us construct the circuit that will reach this state in $O(n^2)$ $\CX$ gates. In order to do that, we will work backwards and reduce signatures $x_0, x_1, \ldots, x_{n-1}$ to the standard basis, performing what is essentially a topologically-restricted Gaussian elimination. We will do that in two phases:

\begin{enumerate}
    \item Guarantee that $\forall k,l, k<l : \sprod{v_l(t), e_k}=0$ and $\forall k: \sprod{v_k(t), e_k}=1$,
    \item Guarantee that $\forall k,l,k\neq l : \sprod{v_l(t), e_k}=0$ and $\forall k: \sprod{v_k(t), e_k}=1$.
\end{enumerate}

We perform the phase (1) inductively in the order of all $k$ from $0$ to $n - 1$.
Let us fix $k\in\{0,\ldots,n-1\}$ and assume that we have placed $t_0$ gates so far, i.e. that the time is $t_0$.
We claim that it is impossible for all $l\geq k$ to satisfy $\sprod{v_l(t_0), e_k}=0$.
Otherwise, the signatures of the last $n - k$ wires would span the space of dimension at most $n - k - 1$, since $\sprod{v_p(t_0), e_q}=0$ for all $p\geq k$ and $q\leq k$ ($q<k$ due to the inductive construction, and $q=k$ by our supposition).
Thus, some signatures would be linearly dependent, which is impossible.
Therefore we can select $i_1 \coloneqq \min \{ j \in \{k, \ldots, n-1\} : \sprod{v_j(t_0), e_k}=1 \}$, since the set is not empty.

Our next goal is to have $\sprod{v_k(t_1), e_k}=1$ at some time $t_1\geq t_0$.
If $i_1=k$, then we do not need to do anything and have $t_1=t_0$.
Assume that $i_1 > k$.
Then we append to our solution an array of $\CX$es of a form $[\CX(i_1-s, i_1-s-1): s\in\{0, \ldots, i_1-k-1\}]$ in the order of increasing $s$.
%$\CX(i_1, i_1 - 1), \CX(i - 1, i - 2), \ldots, \CX(i - (i - k + 1), i - (i - k))$.
The last $\CX$ (placed at time $t_1=t_0+i_1-k-1$) is $\CX(k + 1, k)$.
By the definition of $i_1$, due to the execution of $\CX(i_1 - s, i_1 - s - 1)$, we have $\sprod{e_k, v_{i_1 - s-1}(t_1)} = 1$ --- in particular, $\sprod{v_k(t_1), e_k} = 1$.
This phase requires $i_1 - k$ $\CX$ gates.

Now, we are guaranteed that for $s \in \{k, k + 1, \ldots, i_1\}$ we have $\sprod{v_s(t_1), e_k} = 1$.
Next we pick $i_2 \coloneqq \max \{ j\in\{i_1,\ldots,n-1\} : \sprod{v_j(t_1), e_k}=1 \}$.
We now append the following array of $\CX$es: $[\CX(j, j + 1) : i \leq j < i_2 \land \sprod{v_{j + 1}(t_1), e_k}=0]$ in the order of increasing $j$.
After this operation, i.e. at some time $t_2\geq t_1$, we are guaranteed that for $s \in \{k, k + 1, \ldots, i_2 \}$ we have $\sprod{v_s(t_2), e_k} = 1$.
It thus suffices to apply a sequence of $\CX$es removing the bit from wires $s > k$.
We apply the following $\CX$es in a sequence: $\CX(i_2 - 1, i_2), \CX(i_2 - 2, i_2 - 1), \ldots, \CX(i_2 - (i_2 - k), i_2 - (i_2 - k - 1))$.
Now, i.e. at time $t_3\geq t_2$, we are guaranteed that for $s > k$ we have $\sprod{v_s(t_3), e_k} = 0$ and $\sprod{v_k(t_3), e_k} = 1$.

This finishes the first phase (for wire $k$) successfully.
During this phase, we introduce the following number of $\CX$es:
\begin{itemize}
    \item Between times $t_0$ and $t_1$: $i_1-k$,
    \item Between times $t_1$ and $t_2$: at most $\max(0, i_2-i_1-1)\leq i_2-i_1$,
    \item Between times $t_2$ and $t_3$: $i_2-k$.
\end{itemize}
By summing the above numbers we get that we place at most $2i_2-2k\leq 2n-2k$ $\CX$ gates for wire $k$.
The total number of $\CX$ gates introduced in phase (1) is thus at most $\sum_{k = 0}^{n - 1} 2n-2k = 2n^2 - n(n - 1) = n^2 + n$.

In the second phase we iterate through wires in the order $n - 1, \ldots, 0$.
Let us again consider a wire $k$ at a time $t_4$.
If for every wire $j<k$ we have $\sprod{v_j(t_4), e_k}=0$, then we are done.
Otherwise, we pick $i_3 \coloneqq \min \{ j\in\{0,\ldots,k-1\} : \sprod{v_j(t_4), e_k} = 1 \}$.
First we fill the gaps in a process similar to the one between times $t_1$ and $t_2$, but with $\CX$es directed upwards.
That way, after this sequence of $\CX$ operations, i.e. at time $t_5\geq t_4$, we have $\sprod{v_s(t_5), e_k} = 1$ for $s \in \{i_3, \ldots, k\}$ and $\sprod{v_s(t_4), e_k} = 0$ otherwise.
Then, we ``clear'' all wires below $k$, similarly to the way we did between times $t_2$ and $t_3$ with the exception that now all the $\CX$ gates are directed upwards.
After this step, i.e. at time $t_6\geq t_5$, we have $v_k(t_6)=e_k$ and $\sprod{v_s(t_6), e_k} = 0$ for $s \neq k$.
Note that due to the fact that all the gates in the (2) phase are directed upwards and target qubits are always less than $k$, we maintain the condition from phase (1).

This finishes the construction of the phase (2).
During this phase, we introduce the following number of $\CX$es:
\begin{itemize}
    \item Between times $t_4$ and $t_5$: at most $\max(0, k-i_3-1)\leq k-i_3$,
    \item Between times $t_5$ and $t_6$: $k-i_3$.
\end{itemize}
The total number of operations in this phase is thus bounded by $\sum_{k=0}^{n-1} 2k = n(n-1)$ --- therefore the total number of CX gates is at most $n^2 + n + n^2 - n = 2 n^2 = O(n^2)$.

To reach state $x_0, x_1, \ldots, x_{n-1}$ from the identity state, we simply reverse the order of $\CX$es.
\end{proof}

\begin{theorem}
The shortest solution to NPA problem on linear topology has $\Theta(2^n)$ $\CX$ gates.
\end{theorem}

\begin{proof}
Take any shortest solution $S$ to SPA.
By \Cref{thm:linear-spa-upper-bound}, any such solution has $\Theta(2^n)$ $\CX$ gates.
Take any shortest solution $T$ to reachability problem for signatures $v_k(\len(S))$ for all wires $k$.
By \Cref{thm:reachability-upper-bound}, we know that $\len(T) = O(n^2)$.
Now, simply extending $S$ by $\operatorname{reverse}(T)$ yields a circuit of length $\Theta(2^n) + O(n^2) = \Theta(2^n)$ that is a solution to NPA problem.
\end{proof}

\subsection{Circular topology} We will now briefly consider a variant of the problem where we allow gates
$$ \CX(j, j \pm 1 \textrm{ mod } n) \textrm{ for } 0 \leq j \leq n-1\text{,} $$
where $n$ is the number of qubits. While this topology differs from the linear topology by only 2 gates, it allows for an elegant and almost optimal solution for many $n$.

\begin{theorem} \label{circle_topo_thm}
Let $1 \leq k < n$ be coprime, and let $1 \leq l < n$ be such that $kl \equiv -1 \mod n$. Furthermore, assume that that the polynomial $f(x)\coloneqq x^n + x^l + 1$ is primitive over $\mathbb{F}_2$, i.e. $f$ is minimal and its roots have order $2^n - 1$ in the multiplicative group of $\mathbb{F}_{2^n}$. Then, a circuit consisting of gates
$$ \CX(kj \textrm{ mod } n, kj+1 \textrm{ mod } n), \textrm{ for } 0 \leq j \leq 2^n-2, $$
solves the WPA problem for $n$ qubits in circular topology.
\end{theorem}

\begin{proof}
%Once again, we view the wire states as $n$-dimensional $0-1$ vectors.
Let us interpret the state of the whole system (i.e., the sequence of $n$ wire signatures at a given time) as a $n \times n$ matrix, where $j$-th row is the signature of $j$-th wire.
The matrix multiplications are always performed modulo 2.

Consider a $n\times n$ matrix $A \coloneqq C^k B$, where
\[
B = 
\begin{bmatrix}
 1 &  0 & 0  & \cdots      & 0 & 0 \\
 1 &  1 & 0  & \cdots      & 0 & 0 \\
 0 &  0 & 1  & \cdots      & 0 & 0 \\
 \vdots & \vdots & \vdots & \vdots & \vdots & \vdots \\
 0 &  0 & 0  & \cdots      & 1 & 0 \\
 0 &  0 & 0  & \cdots      & 0 & 1
\end{bmatrix}
\text{and $C=$}
\begin{bmatrix}
 0 &  1 & 0 &  \cdots      & 0 & 0 \\
 0 &  0 & 1  & \cdots      & 0 & 0 \\
 0 &  0 & 0  & \cdots      & 0 & 0 \\
 \vdots & \vdots & \vdots & \vdots & \vdots & \vdots \\
 0 &  0 & 0  & \cdots      & 0 & 1 \\
 1 &  0 & 0  & \cdots      & 0 & 0
\end{bmatrix}.
\]
In other words, $B$ adds first term to the second term and $C$ performs a single cyclic shift of all terms.
Let $X$ be a matrix with rows $w_0, w_1, \dots w_{n-1}$. Then rows of $AX$ are equal to 
\[
w_k, w_{k+1}, \dots, w_{n-1}, w_0, w_0+w_1, w_2, \dots, w_{k-1}.
\]
Now, let us consider a circuit described in the theorem and fix a time $0 \leq j \leq 2^n - 1$.
%Let $v_0, v_1, \dots v_{n-1}$ be the wire states after executing initial $j$ gates of the circuit described in the theorem.
By the previous observation and from a simple inductive argument it follows that the matrix $A^j$ has rows
\[
v_{kj \textrm{ mod } n}(j), v_{kj+1 \textrm{ mod } n}(j), \dots, v_{kj+n-1 \textrm{ mod } n}(j).
\]
To prove the theorem, it is thus sufficient to show that:
\begin{enumerate}[(i)]
    \item $A^{2^n - 1} = I$,
    \item For every nonzero $n$-dimensional $0-1$ vector $v$, there exists $0 \leq j \leq 2^n - 1$ such that $v$ is one of the rows of matrix $A^j$.
\end{enumerate}
We compute the characteristic polynomial\footnote{While the formula for this characteristic polynomial is not immediately obvious, the computation is left to the reader to keep the proof short.} of $A$:
%\todo{this is not obvious, but I think it's best to skip the computation anyway to keep the proof short}
\[
\det(xI - A) = x^n - x^l - 1.
\]
By assumption, it follows that $f(x) = x^n + x^l + 1$ is the minimal characteristic polynomial of matrix $A$, viewed over the field $\mathbb{F}_2$. Let $\mathbb{F}_2(\alpha)$ be an extension of $\mathbb{F}_2$ by a root of $f$. There is an injective homomorphism $\phi$ from the field $\mathbb{F}_2(\alpha)$ to the ring of $n \times n$ matrices over $\mathbb{F}_2$, given by
\begin{align*}
\phi(0) & = \textrm{zero matrix,} \\
\phi(1) & = \textrm{identity matrix,} \\
\phi(\alpha) & = A.
\end{align*}
Since $\alpha^{2^n - 1} = 1$, (i) must hold. Now, consider the matrices 
\begin{equation} \label{matrix_sequence_1}
I, A, A^2, \dots, A^{2^n - 2}.
\end{equation}
We claim that top rows of those matrices are all distinct: pick any $0 \leq i \neq j \leq 2^n - 2$. Then $A^i - A^j = \phi(\alpha^i - \alpha^j)$ is nonzero, hence invertible. This implies that the top row of $A^i - A^j$ is nonzero, and thus top rows of $A^i$ and $A^j$ differ. Since there are $2^n - 1$ nonzero $0-1$ vectors of length $n$, and $2^n - 1$ matrices in (\ref{matrix_sequence_1}), it follows that each vector appears at the top of one of the matrices. This proves (ii).
\end{proof}

\begin{corollary}
    Let $n > 1$ be such that there exists a primitive trinomial $x^n + x^l + 1$ over $\mathbb{F}_2$. Then the WPA problem for $n$ qubits in circle topology has a solution in $2^n - 1$ $\CX$ gates.
\end{corollary}
\begin{proof}
The primitivity of $x^n + x^l + 1$ implies that $n$, $l$ are coprime. Thus we can pick a $k$ that satisfies the assumptions of the previous theorem.
\end{proof}
% \todo{remark that this is optimal solution for SPA}

From the above corollary, we can easily construct a solution to $\SPA$ with the optimal number of $\CX$ gates for certain values of the number of qubits $n$.

\begin{corollary} \label{cor:spa_circle_optimal}
    Suppose $n, k$ satisfy the conditions of Theorem \ref{circle_topo_thm}. Then a circuit consisting of gates
    $$ \CX(kj \textrm{ mod } n, kj+1 \textrm{ mod } n), \textrm{ for } 0 \leq j \leq 2^n-n-2 $$
    solves the SPA problem for $n$ qubits in circular topology.
\end{corollary}
\begin{proof}
Consider the circuit from Theorem \ref{circle_topo_thm}. We claim that during last $n$ steps, no new wire signatures have been visited, and hence after removing the last $n$ CX gates, what remains is a solution to SPA problem. Indeed, since $n, k$ are coprime, every wire $j$ has been the target wire of a CX exactly once during the last $n$ steps. Since the full circuit provides a solution to WPA problem, this CX must have set $j$-th wire to one of the initial states.
\end{proof}

\begin{corollary} \label{cor:circle_not_primitive}
Let $1 \leq k < n$ be coprime, and let $1 \leq l < n$ be such that $kl \equiv -1 \mod n$. Furthermore, assume that that the polynomial $f(x)\coloneqq x^n + x^l + 1$ is irreducible over $\mathbb{F}_2$, and its roots have order $r$ in the multiplicative group of $\mathbb{F}_{2^n}$. Let $q = (2^n - 1) / r$. Then the WPA problem for $n$ qubits in circular topology can be solved in at most
$$ 2^n + (q - 1) O(n^2) $$
CX gates.
\end{corollary}
\begin{proof}
Define $A$ as in proof of Theorem \ref{circle_topo_thm}. Again, the set
$$ F = \{ \epsilon_0 + \epsilon_1 A + \dots + \epsilon_{n-1} A^{n-1} : \epsilon_0, \dots, \epsilon_{n-1} \in \{0, 1\} \},
$$ 
with matrix multiplication and addition, is a field with $2^n$ elements. Pick a primitive element $G$ of this field (that is, an element with maximal order $2^n - 1$ in multiplicative group of $F$. Since there are $\varphi(2^n - 1)$ such elements, we can just pick elements of $F$ at random, until we find a primitive element). We claim the following procedure solves the SPA problem:

For $s = 0, \dots, q - 1$:
\begin{itemize}
    \item prepare the wire states corresponding to rows of $G^s$. By Theorem \ref{thm:reachability-upper-bound} this can be done in $O(n^2)$ CX gates. For $s=0$, we don't need to do anything;
    \item apply $CX(kj \text{ mod } n, kj+1 \text{ mod } n)$ for $j=0, \dots, r-1$.
\end{itemize}

Note that the matrices $I, A, A^2, \dots, A^{r-1}$ are equal (in some order) to $I, G^q, G^{2q}, \dots, G^{(r-1)q}$. Thus at $s$-th step of the aforementioned procedure, we see wire states which correspond (up to cyclic shift) to matrices $G^s, G^{s+q}, \dots G^{s+(r-1)q}$. Throughout the whole procedure, we will see wire states corresponding to every $G^j$, $0 \leq j < 2^n-1$. Reasoning as in proof \ref{circle_topo_thm}, every nonzero vector appears as a row of one of those matrices, and thus the wire state described by this vector has been visited.
\end{proof}

The constructions outlined above are applicable for a wide range of $n$ values. Let us focus on $2 \leq n \leq 20$ and SPA variant:
\begin{itemize}
    \item For $n=2, 3, 4, 5, 6, 7, 9, 10, 11, 15, 17, 18, 20, \ldots$ a primitive trinomial over $\mathbb{F}_2$ exists, and thus by Corollary \ref{cor:spa_circle_optimal} we obtain an optimal solution in $2^n - n - 1$ steps;
    \item For $n=12$, one can choose $k=7$ and $l=5$ satisfy the assumptions of Corollary \ref{cor:circle_not_primitive}, with $q=5$. For $n=14$, the same corollary can be applied, with $k=11, l=5, q=3$;
    \item For $n=8, 13, 16, 19, \ldots$ there are no irreducible trinomials over $\mathbb{F}_2$.
\end{itemize}
See, for example, \cite{ZIERLER1968541} for tables of irreducible and primitive trinomials over $\mathbb{F}_2$.

\section{Adaptive optimization}

There are two main methods to reduce the number of $\CX$ gates in some special cases. The first one analyses the angles of the operator and adjusts visited $v_k(t)$ accordingly. The second one fixes available free parameters to generate appropriate symmetries.

\subsection{Symmetry discovery}

Given a particular diagonal operator with angles $(\alpha_0, \alpha_1, \ldots, \alpha_{2^n - 1})$ with $\alpha_0 = 0$, we first compute the corresponding $\hat \theta$ as presented in \Cref{sec:hat-theta}. Define $S = \{ s \in \mathbb{F}_2^n : \theta_s = 0\}$. The condition may be replaced with $\abs{\theta_s} < \epsilon$ if we allow for approximate decomposition as well. Notice, that $P(0) = I$. In other words, we do not have to visit vector $v_k(t) = s$ corresponding to $\theta_s$ for $s \in S$.

We may thus devise a $\CX$ enumeration scheme through all vectors from $\mathbb{F}_2^n \setminus (S\cup\{0\})$, which may require less gates than the solution in the general case. In this paper, we do not provide a method for finding such schemes. Note, that method similar to $\GRAY$ circuit may be applied with some modification --- namely, we enumerate $\mathbb{F}_2^n \setminus (S\cup\{0\})$ (possibly enumerating some part of $S$ if necessary) on the middle wire using circuit $V$ with appropriate parameters.

Note that for some values of $\hat\theta$ there may be no enumeration scheme omitting entire $S$. For example, let $n = 3$, $\theta_{011} = \theta_{101} = \theta_{110} = 0$ and $\theta_{111} \neq 0$. We cannot visit $v_k(t) = 111$ not visiting at least one of $011, 101$ or $110$.
We can still benefit from it by omitting the remaining two $v_k(t)$, though.

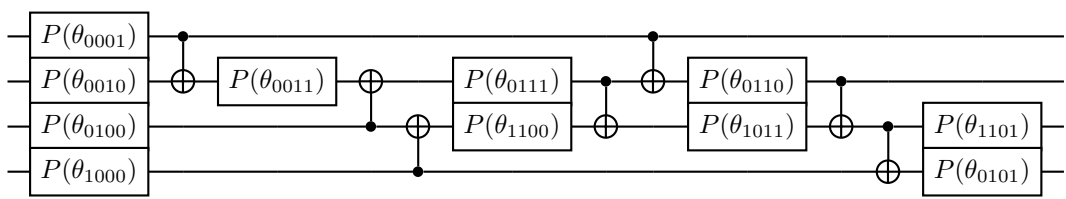
\begin{figure}[!h]
\begin{quantikz}[row sep={0.6cm,between origins}, column sep=0.3cm]
\lstick{} &\gate{P(\theta_{0001})} & \ctrl{1} & \qw                    & \qw      & \qw      & \qw      & \qw& \ctrl{1} & \qw      & \qw      & \qw & \qw & \qw \\
\lstick{} &\gate{P(\theta_{0010})} & \targ{}  &\gate{P(\theta_{0011})} & \targ{}  & \qw      & \gate{P(\theta_{0111})} &\ctrl{1} & \targ{} & \gate{P(\theta_{0110})}  & \ctrl{1} & \qw  & \qw    & \qw  \\
\lstick{} &\gate{P(\theta_{0100})} & \qw      & \qw                    & \ctrl{-1}& \targ{}  & \gate{P(\theta_{1100})} &\targ{}  & \qw   & \gate{P(\theta_{1011})}   & \targ{}  & \ctrl{1} & \gate{P(\theta_{1101})} &\qw  \\
\lstick{} &\gate{P(\theta_{1000})} & \qw      & \qw                    & \qw      &\ctrl{-1} & \qw & \qw      & \qw      & \qw  & \qw    & \targ{}  & \gate{P(\theta_{0101})} & \qw 
\end{quantikz}
\caption{Example solution for SPA on linear topology for $n = 4$ given $\theta_{1001} = \theta_{1010} = \theta_{1110} = \theta_{1111} = 0$. Corresponding $v_k(t)$ are never visited and the solution requires only 7 $\CX$ gates as opposed to 11 required to enumerate all $\mathbb{F}_2^4 \setminus \{0\}$ in the general case.}
\end{figure}

\subsection{Symmetry generation}

Given a non-empty set of basis states $S$ and an operator $D$ diagonal after restriction to $S$ (we assume that $\ket{0} \not \in S$ for clarity). Let us assume that the operator will be applied to a quantum state that is a superposition consisting solely of elements of $S$. Then, instead of implementing $D$, we may implement a diagonal operator $D'$, such that for $s \in S$, we have $D' \cdot s = D \cdot s$ and for $s' \not \in S$ we treat $\alpha_{s'}$ as a formal phase, used to compute (symbolically) $\hat \theta$. Now, we may select $\alpha_{s'}$ for those elements in such a way, that exactly $2^n - 1 - \card{S}$ elements of $\hat \theta$ are 0. There are many different parametrizations of those formal alphas yielding such highly symmetric $D'$, which may lead to different minimal $\CX$ count.

Obviously, after selecting alphas in $D'$, we execute standard symmetry discovery optimization.

\appendix
\section{Computation of \texorpdfstring{$\hat\theta$}{theta}} \label{sec:hat-theta}

To compute value of $\hat\theta$ in the problem description, we first identify $\theta_v$ (where $v \in \mathbb{F}_2^n$) with $\theta_s$ if the $s$-th phase gate in the circuit from \Cref{thm:simplify-phase-gates} is executed on the $k$-th wire at moment $t$ and $v = v_k(t)$. We also note that $P(\theta_v)$ introduces the indicated phase to every state of the computational basis $\ket b$ such that $v \cdot b = 1$.

Thus, the total phase applied to the computational basis state $\ket b$ is
$$
    \alpha_b = \sum_{\substack{v \in \mathbb{F}_2^n \\ v \cdot b = 1}} \theta_v\text{.}
$$
Thus, $\hat\theta = M^{-1} \hat\alpha$, where $\hat\alpha = (\alpha_1, \ldots, \alpha_{2^n - 1})$ and $M = [m_{ij}]$ is a $2^n - 1 \times 2^n - 1$ matrix such that $m_{ij} = i \cdot j$ where $\cdot$ is a $\mathbb F _2^n$ inner product and $i, j$ binary representation are read as vectors from $\mathbb F _2^n$.

For computational purposes, we may notice the strong similarities between matrix $M$ and Hadamard matrix $H = [h_{ij}]$ having $h_{ij} = (-1)^{i \cdot j}$, where $\cdot$ is a $\mathbb F _2^n$ inner product. Let us also introduce matrix $J \coloneqq [1]_{ij}$. Then $M$ can be obtained from matrix $\frac{1}{2} (J - H)$ by removing the first row and column. Let us write $\bar \theta = (\theta_0, \theta_1, \ldots, \theta_{2^n - 1})$ and $\bar \alpha = (\alpha_0, \alpha_1, \ldots, \alpha_{2^n - 1})$. Note that $\theta_0$ is a formal vector element and $\alpha_0 = 0$. Then, since $M\hat\theta = \hat \alpha$, we have that $\frac{1}{2} (J - H) \bar\theta = \bar \alpha$.
Multiplying this on the left hand side by $-2^{-n+1} H$, we obtain $(-U + I) \bar\theta = -2^{-n+1} H\bar \alpha$, where $U = [u_{ij}]$ and $u_{ij}$ is 1 if $i = 0$ and 0 otherwise. Elements of the left hand side's product are, starting from $k = 1$, precisely $\theta_k$. Thus, $\hat \theta$ can be obtained by performing the fast Walsh-Hadamard Transform on $\bar \alpha$, skipping the first element of the resulting vector, and dividing it by $-2^{n - 1}$. This algorithm allows to compute $\hat \theta$ in $O(n 2^n)$ time.

\section{Optimal circuits}
In the paper we presented several circuit families with simple constructions and provable scalability.
Circuits included in those families may however not be optimal for particular $n$.
In the \Cref{tab:optimal-cx_count} we present upper-bound for the $\CX$ gate count in the shortest solutions for given variants.
Table entries given as number rather than inequalities are also lower-bounds for the particular problem setup and thus are optimal.

% \todo{porownac tabelke z algorytmami z pracki}

\begin{center}
\begin{table}[h!]
\begin{tabular}{ |c||c|c|c||c|c|c||c|c|c| } 
\hline
Wire count & \multicolumn{3}{|c||}{Fully-connected} & \multicolumn{3}{|c||}{Linear} & \multicolumn{3}{|c|}{Circular} \\
& SPA & WPA & NPA & SPA & WPA & NPA & SPA & WPA & NPA \\
\hline
\hline 
2 & 1 & 2 & 2 & 1 & 2 & 2 & 1 & 2 & 2 \\  
\hline
3 & 4 & 6 & 6 & 5 & 7 & 8 & 4 & 6 & 6 \\  
\hline
4 & 11 & 14 & 14 & 14 & 17 & 18 & 11 & 14 & 16 \\  
\hline
5 & 26 & 30 & 30 & 31 & 37 & 40 & 26 & $\leq 31$ & $\leq 35$ \\  
\hline
6 & 57 & 62 & 62 & $\leq 69$ & $\leq 80$ & $\leq 86$ & 57 & $\leq 63$ & $\leq 73$ \\  
\hline
7 & 120 & 126 & 126 & $\leq 150$ & $\leq 172$ & $\leq 178$ & 120 & $\leq 127$ & $\leq 147$ \\  
\hline
8 & 247 & 254 & 254 & $\leq 327$ & $\leq 355$ & $\leq 390$ & $\leq 267$ & $\leq 299$ & $\leq 317$ \\  
\hline
\end{tabular}
\caption{Number of $\CX$ gates in the shortest solution to different problem variants on analysed topologies. Note that the fully-connected topology contains theoretical lower bounds for all other topologies.}
\label{tab:optimal-cx_count}
\end{table}
\end{center}

Some of the entries in the table above were obtained using various heuristics and optimization techniques. The constructions presented in the paper are usually worse than the optimized ones by a surprisingly small factor, with the advantage of not requiring any additional computational power. Besides the obvious cases of fully-connected topology and most of circular $\SPA$ (which are both solved optimally in the paper) and $\WPA$ (which can possibly be improved by at most one $\CX$), the constructed circuit $\SPA_8$ for the linear topology requires 409 $\CX$ gates -- which is only about 25\% worse than the best solution obtained by us using extensive heuristic optimization at 327 $\CX$es.
\newpage
\bibliographystyle{acm}
\bibliography{bibliography}

\end{document}